\documentclass[letterpaper, 10 pt, conference]{ieeeconf}  

\IEEEoverridecommandlockouts        

\usepackage{amsfonts}
\usepackage{graphicx}
\usepackage{float}
\usepackage{comment}

\usepackage{amsthm,amssymb} 
\usepackage{amsmath}
\usepackage{comment}
\usepackage{xcolor}
\newtheorem{theorem}{Theorem}
\newtheorem{remark}{Remark}

\newtheorem{lemma}{Lemma}

\newcommand{\x}{\textup{x}}

\newcommand{\w}{\textup{w}}
\newcommand{\ph}{\varphi}
\newcommand{\eps}{\varepsilon}

\newcommand{\rd}{\color{red}}

\title{\LARGE \bf
Lyapunov--Krasovskii functionals for some classes of \\ nonlinear time delay systems}

\author{Gerson Portilla$^{1}$, Irina V. Alexandrova$^{2}$ and Sabine Mondi\'e$^{1}$ 
\thanks{*The work of the first and third authors was supported by project CONACYT A1-S-24796, Mexico. The work of the second author was supported
by the Russian Science Foundation, Project 19-71-00061.}
\thanks{$^{1}$Gerson Portilla and Sabine Mondi\'e are with Department of Automatic Control, CINVESTAV-IPN, 07360 Mexico D.F., Mexico
        {\tt\small {gportilla,smondie}@ctrl.cinvestav.mx}}%
\thanks{$^{2}$Irina V. Alexandrova is with Department of Applied Mathematics and Control Processes,
        St.Petersburg State University, 199034 St.Petersburg, Russia
        {\tt\small i.v.aleksandrova@spbu.ru}}%

}

\begin{document}
\maketitle
\thispagestyle{empty}
\pagestyle{empty}

\begin{abstract}
In this contribution, we study an homogeneous class of nonlinear time delay systems with time-varying perturbations. Using the Lyapunov-Krasovskii approach, we introduce a functional that leads to perturbation conditions matching those obtained previously in the Razumikhin framework. The functionals are applied to the estimation of the domain of attraction and of the system solutions. An illustrative example is given.
\end{abstract}

\section{Introduction}
The homogeneous approximation is the top choice for the stability analysis for nonlinear systems when the first approximation is zero. This strategy along with the Lyapunov framework have produced several results on stability \cite{aleksandrov2012asymptotic,aleksandrov2014delay,efimov2016}, robustness \cite{rosier1992homogeneous}  and controller design \cite{hermes1991homogeneous}.  The Lyapunov-Razumikhin approach has had a key role in the development of the main results on homogeneous approximations of time-delay systems: delay-independent stability \cite{aleksandrov2012asymptotic,aleksandrov2014delay,efimov2016} and applications such as estimates of the solutions \cite{aleksandrov2012asymptotic,aleksandrov2016asymptotic} and of the region of attraction \cite{aleksandrov2014delay}. The Lyapunov-Razumikhin framework allowed establishing an outstanding result: if the delay-free system is asymptotically stable, then the trivial solution of the homogeneous delay system is asymptotically stable for all delays when the homogeneity degree is strictly greater than one \cite{aleksandrov2012asymptotic}.\\ A common practice is to use the Lyapunov function of the delay-free system to determine the stability of the system in presence of delays, as proposed in the Lyapunov-Razumikhin approach.  The asymptotic stability of delay free systems with time-varying perturbations has been addressed in \cite{aleksandrov1996stability} in the framework of the Lyapunov direct method.  There, a method for constructing a Lyapunov functions was given, and conditions insuring the system subject to such time varying perturbations remains asymptotically stable was presented. In \cite{aleksandrov2012asymptotic}, the above Lyapunov function for a delay free system was used in the Lyapunov-Razumikhin framework to determine  conditions for the asymptotic stability of time-delay systems with perturbations.\\
In recent years, the study, in the Lyapunov-Krasovskii framework, of homogeneous systems with delays  has received some attention thanks to the Lyapunov-Krasovskii functional for systems of homogeneity degree strictly greater than one introduced in
\cite{Voronezh, alexandrova2019lyapunov}. Its construction, whose starting point is the homogeneous delay free system
Lyapunov function, is inspired by the so-called complete type
functionals approach \cite{kharitonov2013time}.\\
In \cite{aleksandrov2012asymptotic,aleksandrov2016asymptotic}, 
the following class of nonlinear systems with time-varying perturbations
\begin{equation}\label{eq:system_zero_mean_value}
    \dot{x}(t)=f(x(t),x(t-h))+B(t)Q(x(t),x(t-h))
\end{equation}
is studied.
Here, $x(t)\in \mathbb{R}^n,$ $h>0$ is a constant delay, $f(\x_1,\x_2)$ is a homogeneous vector function of degree $\mu>1,$ that is
$$
f(c\x_1,c\x_2)=c^\mu f(\x_1,\x_2)\quad\forall c>0,\;\x_1,\x_2\in \mathbb{R}^n,
$$
satisfies the Lipschitz condition and is continuously differentiable with respect to $\x_2,$
$B(t)$ is a matrix with continuous and bounded elements, $\|B(t)\|\leq \hat{b},$ and the vector function $Q(\x_1,\x_2)$ is continuously differentiable and satisfies
\begin{align}\label{eq:bound_Q}
    \|Q(\x_1,\x_2)\|&\leq p_1\|\x_1\|^{\sigma}+p_2\|\x_2\|^{\sigma},\quad p_1,p_2\geq 0,\\
     \left\|\frac{\partial Q(\x_1,\x_2)}{\partial \x_j}\right\|&\leq q_{j1}\|\x_1\|^{\sigma-1}+q_{j2}\|\x_2\|^{\sigma-1},\notag
\end{align}
$q_{jk}\geq 0,\;j,k=1,2$. The following result is known.

\begin{theorem} \textup{\cite{aleksandrov2012asymptotic,aleksandrov2016asymptotic}} \label{thm_Aleksandrov}
If the delay free system
\begin{equation} \label{eq:delay_free_sys}
    \dot{x}=f(x(t),x(t))
\end{equation}
is asymptotically stable, then the asymptotic stability is preserved for the trivial solution of system \eqref{eq:system_zero_mean_value} with perturbations of the following classes: \\$(a)$ the integral
\begin{align}
\label{eq:int_B}
I(t)=\int_0^t B(s) ds
\end{align}
is bounded and $\sigma>\dfrac{\mu+1}{2};$ \\$(b)$ integral \eqref{eq:int_B} is unbounded but
\begin{align}
\label{eq:integrate_T}
\frac{1}{T}\int_t^{t+T} B(s)ds \xrightarrow[T\to +\infty]{} 0
\end{align}
uniformly with respect to $t\geq 0,$ and $\sigma\geq\mu.$
\end{theorem}
In case $(a),$ the elements of $B(t)$ may describe periodic oscillations with zero mean values, without any restriction on their amplitudes, whereas assumption $(b)$ holds, in particular, when the elements
of $B(t)$ are almost periodic functions with zero mean values. It is also known \cite{aleksandrov2014delay} that if system \eqref{eq:delay_free_sys} is asymptotically stable, then the trivial solution of system \eqref{eq:system_zero_mean_value} with $B(t)\equiv \mathbf{0}$ is asymptotically stable for all values of the delay.
The aim of this paper is to construct the Lyapunov-Krasovskii functionals for system \eqref{eq:system_zero_mean_value}, which on the one hand, allow to verify the results of \cite{aleksandrov2012asymptotic,aleksandrov2016asymptotic}, and on the other hand, are useful in practical applications such as the construction of estimates for the solutions, etc. The functionals are based on the Lyapunov functions for delay free systems constructed by the method introduced in \cite{aleksandrov1996stability}. It is worth mentioning that they are not directly derived from an application of the construction \cite{alexandrova2019lyapunov} with the Lyapunov function \cite{aleksandrov1996stability}. \\
The contribution is organised as follows. In Section II, the main properties of homogeneous systems are
introduced. The construction of the functional is presented in Section III. In Section IV, estimates of the attraction region and of the solutions achieved by using this functional are presented. Section V is devoted to an illustrative example. The paper ends with conclusions.
 
\section{Preliminaries}
In this paper, we assume that the initial functions belong to the space of $\mathbb{R}^n$-valued continuous functions on $[-h,0],$ which is denoted by $C_{[-h,0]}.$ This space is endowed with the norm $\|\varphi\|_h=\max_{\theta\in[-h,0]}\|\varphi(\theta)\|,$ where $\|\cdot\|$ stands for the Euclidean norm. Function $x(t,\ph)$ denotes the solution of system \eqref{eq:system_zero_mean_value} with an initial function $\ph,$ whereas $x_t$ is the state of system \eqref{eq:system_zero_mean_value}:
$$
x_t:\quad\theta\to x(t+\theta),\quad \theta\in[-h,0].$$
Assume that the vector function $f(\x_1,\x_2)$ is continuously differentiable with respect to $\x_1.$ Due to the homogeneity, there exist $m_1,m_2,\eta_{11},\eta_{12}\geq 0$ such that
\begin{align}\label{eq:bound_f}
    \|f(\x_1,\x_2)\|&\leq m_1\|\x_1\|^{\mu}+m_2\|\x_2\|^{\mu},\\
    \left\|\frac{\partial f(\x_1,\x_2)}{\partial \x_1}\right\|&\leq \eta_{11}\|\x_1\|^{\mu-1}+\eta_{12}\|\x_2\|^{\mu-1}.\notag
\end{align}
We emphasize that, for the development of the Lyapunov--Krasovskii approach, differentiability of $f$ with respect to $\x_1,$ instead of that with respect to $\x_2$ in the Lyapunov--Razumikhin approach \cite{aleksandrov2012asymptotic,aleksandrov2016asymptotic}, is required.

Throughout the paper, we assume that the delay free system \eqref{eq:delay_free_sys} is asymptotically stable. It is known \cite{zubov1964methods,rosier1992homogeneous} that there exists a positive definite and twice continuously differentiable Lyapunov function $V(\x),$ which is homogeneous of degree $\gamma\geq 2$ and satisfies together with its derivatives the following inequalities
\begin{gather}
    \left(\frac{\partial V(\x)}{\partial \x}\right)^T f(\x,\x)\leq-\mathrm{w}\|\x\|^{\gamma+\mu-1}, \quad \mathrm{w}>0,\notag\\
    \alpha_0\|\x\|^\gamma\leq V(\x)\leq \alpha_1\|\x\|^\gamma,\label{eq:bound_V}\\
\left\|\frac{\partial V(\x)}{\partial \x}\right\|\leq \beta\|\x\|^{\gamma-1},\ \ \left\|\frac{\partial^2 V(\x)}{\partial \x^2}\right\|\leq \psi\|\x\|^{\gamma-2},\notag
\end{gather}
where $\alpha_0,\,\alpha_1,\,\beta,\,\psi>0.$ Given $\alpha>1,$ introduce the set
$$
S_\alpha = \Bigl\{\varphi\in C_{[-h,0]}\Bigl\arrowvert\|\varphi(\theta)\|\leq \alpha\|\varphi(0)\|,\; \theta\in[-h,0]\Bigr\}.
$$
It was shown in \cite{alexandrova2018junction} that to analyse the stability it is enough to compute a lower bound for the Lyapunov--Krasovskii functional on the set $S_\alpha$, and to ensure the negativity of the time derivative along the solutions of system \eqref{eq:system_zero_mean_value}. The set $S_\alpha$ also plays an important role in the construction of the estimates for solutions in Section~IV.

The following constants will be used in the sequel:
\begin{gather*}
m=m_1+m_2,\quad \eta=\eta_{11}+\eta_{12},\quad p=p_1+p_2,\\ q=q_{11}+q_{12},\quad
\kappa_1=\psi m_2 +\beta\eta_{12},\quad \kappa_2=\psi p_2 +\beta q_{12},\\ L_1=\psi m +\beta\eta,\quad L_2=\psi p +\beta q,\\ L_3=L_2 + \beta(q_{21}+q_{22}).
\end{gather*}
\section{Construction of the functional}
In this section, we construct a Lyapunov--Krasovskii functional for system \eqref{eq:system_zero_mean_value} with perturbations either of class $(a)$ or $(b).$ Following \cite{aleksandrov2012asymptotic},  we introduce the integral
\begin{equation}\label{eq:L_varepsilon}
    L(t,\varepsilon)=\int_0^{t+h} e^{-\varepsilon(t+h-s)}B(s)ds,
\end{equation}
where $\eps=0$ in case $(a)$, and $\eps>0$ in case $(b).$
It is known \cite{fink2006almost} that in case $(b)$ integral \eqref{eq:L_varepsilon} satisfies the following property: there exists a function $\omega(\varepsilon)$ such that $\omega(\varepsilon)\rightarrow 0$ as $\varepsilon\rightarrow 0,$ and
\begin{equation*}
    \eps\|L(t,\varepsilon)\|\leq \omega(\varepsilon)
\end{equation*}
for all $t\geq 0.$ In case $(a)$, there is $l_0>0$ such that $\|L(t,0)\|\leq l_0,$ thus we set $\omega(\varepsilon)=l_0 \eps.$\\
Inspired by the Lyapunov function for system \eqref{eq:system_zero_mean_value} presented in \cite{aleksandrov2012asymptotic} and using the construction of Lyapunov-Krasovskii functionals for homogeneous systems introduced in \cite{Voronezh,alexandrova2019lyapunov}, we propose the following  functional for system \eqref{eq:system_zero_mean_value}:
\begin{align}\label{eq:functional_zero_mean_2}
&v(t,\varphi)=V(\varphi(0))+\left.\left(\frac{\partial V(\x)}{\partial \x}\right)^T \right|_{\x=\varphi(0)}\\
&\times \Biggl(\int_{-h}^{0}\Bigl(f(\varphi(0),\varphi(\theta))+B(t+\theta+h)Q(\varphi(0),\varphi(\theta))\Bigr)d\theta\notag\\&-L(t,\varepsilon)Q(\varphi(0),\varphi(0))\Biggr)\notag\\&+\int_{-h}^{0}(\mathrm{w_1}+(h+\theta)\mathrm{w_2})\|\varphi(\theta)\|^{\gamma+\mu-1}d\theta.\notag
\end{align}
Here, $\mathrm{w}_1, \mathrm{w}_2>0$ are such that \mbox{$\mathrm{w}_0=\mathrm{w}-\mathrm{w}_1-h\mathrm{w}_2>0$.} Notice that functional \eqref{eq:functional_zero_mean_2} does not coincide with the construction from \cite{alexandrova2019lyapunov} applied with the Lyapunov function of \cite{aleksandrov2012asymptotic}.
We prove that the Lyapunov--Krasovskii functional \eqref{eq:functional_zero_mean_2} is suitable for the stability analysis of system \eqref{eq:system_zero_mean_value} below.
\begin{lemma} \label{lemma_bound_dot_v_pertur_zero_mean}
There exist $\delta>0$ and $c_0,c_1,c_2>0$ such that the time derivative of functional \eqref{eq:functional_zero_mean_2} along the solutions of system \eqref{eq:system_zero_mean_value}, admits a bound of the form
\begin{gather}\label{eq:bound_dot_v_pertur_zero_mean}
\frac{dv(t,x_t)}{dt}\leq -c_0\|x(t)\|^{\gamma+\mu-1}-c_1\|x(t-h)\|^{\gamma+\mu-1}\\-c_2\int_{-h}^{0}\|x(t+\theta)\|^{\gamma+\mu-1}d\theta, \notag
\end{gather}
in the neighbourhood $\|x_t\|_h\leq \delta.$
\end{lemma}
\begin{proof}
Differentiating each of the three summands of \eqref{eq:functional_zero_mean_2} along the solutions of \eqref{eq:system_zero_mean_value}, we obtain

\begin{gather*}
\frac{dv(t,x_t)}{dt} = -\mathrm{w}_0\|x(t)\|^{\gamma+\mu-1}-\mathrm{w}_1\|x(t-h)\|^{\gamma+\mu-1}\\-\mathrm{w}_2\int_{-h}^{0}\|x(t+\theta)\|^{\gamma+\mu-1}d\theta+\sum_{j=1}^{5}\Lambda_j,
\end{gather*}
where
\begin{align*}
\Lambda_1 &= \left(\frac{\partial V(\mathrm{x})}{\partial \mathrm{x}}\right)^T\varepsilon L(t,\varepsilon)Q(\mathrm{x},\mathrm{x})\biggl|_{\x=x(t)},\\
\Lambda_2 &= \left(\frac{\partial V(\mathrm{x})}{\partial \mathrm{x}}\right)^T \\&\times\int_{t-h}^{t}\left[\frac{\partial f(\mathrm{x},x(s))}{\partial \mathrm{x}}+B(s+h)\frac{\partial Q(\x,x(s))}{\partial \x}\right]ds \Biggl|_{\x=x(t)}\\
&\times\Bigl(f(x(t),x(t-h))+B(t)Q(x(t),x(t-h))\Bigr),\\
\Lambda_3 &= - \left(\frac{\partial V(\mathrm{x}_1)}{\partial \mathrm{x}_1}\right)^T L(t,\varepsilon)\sum_{j=1}^2\frac{\partial Q(\x_1,\x_2)}{\partial \x_j}\Biggl|_{\x_1=\x_2=x(t)}
\\&\times\Bigl(f(x(t),x(t-h))+B(t)Q(x(t),x(t-h))\Bigr),\\
\Lambda_4 &= \Bigl(f(x(t),x(t-h))+B(t)Q(x(t),x(t-h))\Bigr)^T\\
&\times\left(\frac{\partial^2 V(\x)}{\partial \x^2}\right)\biggl|_{\x=x(t)} \int_{-h}^{0}\Bigl(f(x(t),x(t+\theta))\\ &+B(t+\theta+h)Q(x(t),x(t+\theta))\Bigr)d\theta,
\\
\Lambda_5 &= -\Bigl(f(x(t),x(t-h))+B(t)Q(x(t),x(t-h))\Bigr)^T\\&\times\left(\frac{\partial^2 V(\x)}{\partial \x^2}\right)\biggl|_{\x=x(t)} L(t,\varepsilon)Q(x(t),x(t)).
\end{align*}
Next, in concordance with \eqref{eq:bound_Q}, \eqref{eq:bound_f}, \eqref{eq:bound_V} and using the inequality
\begin{gather*}
\gamma_1^{l_1}\gamma_2^{l_2}\gamma_3^{l_3}\leq \gamma_1^{l_1+l_2+l_3}+\gamma_2^{l_1+l_2+l_3}+\gamma_3^{l_1+l_2+l_3},\\
\gamma_1,\gamma_2,\gamma_3\geq0,\quad l_1,l_2,l_3>0,
\end{gather*}
we estimate the terms $\Lambda_j$ separately:
\begin{align*}
|\Lambda_1|&\leq \beta p \omega(\eps) \|x(t)\|^{\gamma+\sigma-1},\\
|\Lambda_2|&\leq \beta m \eta h \|x(t)\|^{\gamma+2\mu-2} + \beta m_2 \eta h \|x(t-h)\|^{\gamma+2\mu-2} \\&+ \beta m \eta_{12}\int_{-h}^0\|x(t+\theta)\|^{\gamma+2\mu-2} d\theta\\
&+ \beta \hat{b} h (q m + p \eta) \|x(t)\|^{\gamma+\mu+\sigma-2} \\&+ \beta \hat{b} h (q m_2 + p_2 \eta)\|x(t-h)\|^{\gamma+\mu+\sigma-2} \\ &+ \beta \hat{b} (q_{12} m + p \eta_{12}) \int_{-h}^0\|x(t+\theta)\|^{\gamma+\mu+\sigma-2} d\theta\\ &+ \beta \hat{b}^2 p q h\|x(t)\|^{\gamma+2\sigma-2} +  \beta \hat{b}^2 p_2 q h \|x(t-h)\|^{\gamma+2\sigma-2}\\&+
\beta \hat{b}^2 p q_{12}
\int_{-h}^0\|x(t+\theta)\|^{\gamma+2\sigma-2} d\theta,
\end{align*}
\begin{align*}
|\Lambda_4|&\leq \psi m^2 h \|x(t)\|^{\gamma+2\mu-2} + \psi m m_2 h \|x(t-h)\|^{\gamma+2\mu-2} 
\\
&+ \psi m m_2\int_{-h}^0\|x(t+\theta)\|^{\gamma+2\mu-2} d\theta\\
&+ 2\hat{b}\psi m p h \|x(t)\|^{\gamma+\mu+\sigma-2} \\&+ \hat{b}\psi (m p_2+m_2 p) h\|x(t-h)\|^{\gamma+\mu+\sigma-2}\\ &+ \hat{b}\psi (m p_2+m_2 p) \int_{-h}^0\|x(t+\theta)\|^{\gamma+\mu+\sigma-2} d\theta\\ &+ \psi \hat{b}^2 p^2 h\|x(t)\|^{\gamma+2\sigma-2} + \psi \hat{b}^2 p p_2 h \|x(t-h)\|^{\gamma+2\sigma-2}\\&+
\psi \hat{b}^2 p p_2 \int_{-h}^0\|x(t+\theta)\|^{\gamma+2\sigma-2} d\theta,\\
|\Lambda_3&+\Lambda_5|\leq L_3\dfrac{\omega(\eps)}{\eps}\Bigl(m \|x(t)\|^{\gamma+\mu+\sigma-2} \\&+ \hat{b} p \|x(t)\|^{\gamma+2\sigma-2} +m_2 \|x(t-h)\|^{\gamma+\mu+\sigma-2} \\&+ \hat{b} p_2 \|x(t-h)\|^{\gamma+2\sigma-2}\Bigr).
\end{align*}
Notice that $\gamma+2\mu-2>\gamma+\mu-1$ due to $\mu>1.$ Moreover, $\sigma>1$ leads to 
\begin{align*}
\gamma+\mu+\sigma-2&>\gamma+\mu-1,\\
\gamma+2\sigma-2&>\gamma+\mu-1
\end{align*}
both in cases $(a)$ and $(b)$. The term $\Lambda_1$ is absent in case $(a)$ when $\eps=0.$ In case $(b),$ the degree $\gamma+\sigma-1\geq \gamma+\mu-1$ whilst the coefficient $\omega(\eps)$ can be done arbitrarily small by the choice of $\eps.$ Hence, bound \eqref{eq:bound_dot_v_pertur_zero_mean} holds in a neighbourhood $\|x_t\|_h\leq \delta$ with
\begin{align*}
c_0&=\w_0-\beta p \omega(\eps) \delta^{\sigma-\mu}- m h L_1\delta^{\mu-1} \\&- \left(\hat{b}h(p L_1+ m L_2)+m L_3\dfrac{\omega(\eps)}{\eps}\right)\delta^{\sigma-1} \\ &- p \hat{b}\left(\hat{b} h L_2+L_3\dfrac{\omega(\eps)}{\eps} \right) \delta^{2\sigma-\mu-1},\\
c_1&=\w_1- m_2 h L_1\delta^{\mu-1} \\&- \left(\hat{b}h(p_2 L_1+m_2 L_2)+m_2 L_3\dfrac{\omega(\eps)}{\eps}\right)\delta^{\sigma-1}
\\&-p_2 \hat{b}\left(\hat{b} h L_2+L_3\dfrac{\omega(\eps)}{\eps} \right)\delta^{2\sigma-\mu-1},
\\
c_2&=\w_2- m\kappa_1\delta^{\mu-1} -\hat{b}(p\kappa_1+m\kappa_2) \delta^{\sigma-1} \\&- p \hat{b}^2 \kappa_2 \delta^{2\sigma-\mu-1}.    
\end{align*}
Here, $\delta>0$ is chosen in such a way that the constants $c_0,c_1,c_2$ are positive. Notice that if we replace $\omega(\eps)$ with zero and $\omega(\eps)/\eps$ with $l_0$ in the above expressions, we recover case $(a)$.
\end{proof}
\begin{lemma}\label{lemma_lower_bound_functional_pertur_zero_mean_2}
There exist $\delta>0$ and $a_1(\alpha)>0$ such that functional \eqref{eq:functional_zero_mean_2} admits on the set $S_\alpha$ a lower bound of the form
\begin{equation}\label{eq:firs_lower_bound_v_pertur_zero_mean_2}
    v(t,\varphi)\geq a_1(\alpha)\|\varphi(0)\|^\gamma+\mathrm{w}_1\int_{-h}^{0}\|\varphi(\theta)\|^{\gamma+\mu-1} d\theta,
\end{equation}
where $\|\varphi\|_h\leq \delta.$ 
\end{lemma}
\begin{proof}
Using bounds \eqref{eq:bound_Q}, \eqref{eq:bound_f} and \eqref{eq:bound_V}, we estimate the second summand of functional \eqref{eq:functional_zero_mean_2}. First,
\begin{align*}
\Biggl|\left(\frac{\partial V(\x)}{\partial \x}\right)^T &\biggr|_{\x=\varphi(0)} L(t,\varepsilon)Q(\varphi(0),\varphi(0))\Biggr|\\ &\leq \beta p\dfrac{\omega(\eps)}{\eps}\|\varphi(0)\|^{\gamma+\sigma-1}.    
\end{align*}
Second, the remaining term is estimated on the set $S_\alpha:$
\begin{align*}
&\Biggl|\left(\frac{\partial V(\x)}{\partial \x}\right)^T \biggr|_{\x=\varphi(0)} \int_{-h}^{0}\Bigl(f(\varphi(0),\varphi(\theta))+B(t+\theta+h)\\ &\times Q(\varphi(0),\varphi(\theta))\Bigr)d\theta\Biggr|
\leq
\beta h (m_1+m_2\alpha^\mu)\|\varphi(0)\|^{\gamma+\mu-1}
\\
&+ \beta \hat{b} h (p_1+p_2\alpha^\sigma)\|\varphi(0)\|^{\gamma+\sigma-1},\quad \ph\in S_\alpha.
\end{align*}
Hence, the required bound \eqref{eq:firs_lower_bound_v_pertur_zero_mean_2} holds with
\begin{align*}
a_1(\alpha)&=\alpha_0-\beta h (m_1+m_2\alpha^\mu) \delta^{\mu-1} \\&- \left(\beta \hat{b} h (p_1+p_2\alpha^\sigma)+\beta p\dfrac{\omega(\eps)}{\eps}\right)\delta^{\sigma-1},
\end{align*}
where $\delta>0$ is chosen in such a way that $a_1(\alpha)>0.$
\end{proof}
\begin{lemma}\label{lemma_upper_bound}
Functional \eqref{eq:functional_zero_mean_2} admits the upper bound of the form
\begin{equation}\label{eq:upper_bound_1}
    v(t,\varphi)\leq b_0\|\varphi(0)\|^\gamma+b_1\int_{-h}^{0}\|\varphi(\theta)\|^{\gamma} d\theta,\quad b_0,b_1>0,
\end{equation}
in the neighbourhood $\|\varphi\|_h\leq \delta.$ Moreover, there exist $b_2,b_3>0$ such that
\begin{equation} \label{eq:upper_bound_2}
v(t,\varphi)\leq \alpha_1 \|\varphi(0)\|^\gamma + b_2 \|\varphi\|_h^{\gamma+\mu-1} + b_3 \|\varphi\|_h^{\gamma+\sigma-1}.
\end{equation}
\end{lemma}
\begin{proof} Making the estimations in the same way as in Lemma~\ref{lemma_lower_bound_functional_pertur_zero_mean_2}, we arrive at the required bounds with
\begin{align*}
    b_0 &= \alpha_1+\beta m h \delta^{\mu-1} + \beta p\left(\hat{b} h + \dfrac{\omega(\eps)}{\eps}\right) \delta^{\sigma-1},\\
    b_1 &= (\beta m_2 + \w_1 + h \w_2)\delta^{\mu-1} + \beta p_2 \hat{b} \delta^{\sigma-1},\\
    b_2 &= (\beta m + \w_1 + h \w_2)h,\quad
    b_3 = \beta p \left(\hat{b}h + \dfrac{\omega(\eps)}{\eps}\right).
\end{align*}
\end{proof}
Lemmas~\ref{lemma_bound_dot_v_pertur_zero_mean}--\ref{lemma_upper_bound} imply that functional \eqref{eq:functional_zero_mean_2} allows verifying Theorem~\ref{thm_Aleksandrov} \cite{alexandrova2018junction}, with slightly different assumptions on the right-hand sides.
\section{Estimates}
\subsection{Estimates for the attraction region}
Based on the classical ideas \cite{Krasovskii,melchor2007estimates} and the bounds presented in Lemmas~\ref{lemma_bound_dot_v_pertur_zero_mean}--\ref{lemma_upper_bound}, we obtain the following estimate for the attraction region of the trivial solution of system (\ref{eq:system_zero_mean_value}), see \cite{alexandrova2019lyapunov} for the use of the set $S_\alpha$ in this context.
\begin{theorem}\label{thm:attr_region_LK}
Let system \eqref{eq:delay_free_sys} be asymptotically stable and
$\Delta$ be a positive root of equation
\begin{equation}
\label{eq:attraction_LK}
    \alpha_1\Delta^\gamma +b_2 \Delta^{\gamma+\mu-1}+b_3 \Delta^{\gamma+\sigma-1}=a_1(\alpha)\delta^\gamma,
\end{equation}
where $\delta$ is defined in Lemmas~\ref{lemma_bound_dot_v_pertur_zero_mean}--\ref{lemma_lower_bound_functional_pertur_zero_mean_2}. Then, the set of initial functions $\|\varphi\|_h<\Delta$
is the estimate of the attraction region of the trivial solution of \eqref{eq:system_zero_mean_value}.
\end{theorem}
\begin{remark}
It follows from the proof of Theorem~\ref{thm:attr_region_LK} that 
$\|\varphi\|_h<\Delta$ implies $\|x_t\|_h\leq \delta$ for all $t\geq 0.$
\end{remark}
\subsection{Estimates for the solutions}
In \cite{portilla2021weighted}, a novel approach for the calculation of the estimates of the solutions is presented, which combines the use of Lyapunov-Krasovskii functionals with ideas of the Razumikhin framework. In this section, we apply this approach to system (\ref{eq:system_zero_mean_value}) using functional \eqref{eq:functional_zero_mean_2}. 
First, as in \cite{portilla2021weighted,portilla2020comparison}, we connect functional $v(t,\varphi)$ with its time derivative based on Lemmas~\ref{lemma_bound_dot_v_pertur_zero_mean} and \ref{lemma_upper_bound}:
\begin{equation}\label{eq:connect_v_dot_v} 
    \frac{dv(t,x_t)}{dt} \leq -\rho v(t,x_t)^{\frac{\gamma+\mu-1}{\gamma}},\quad t\geq0,
\end{equation}
where
$$
\rho=\frac{c}{b^{\frac{\gamma+\mu-1}{\gamma}}\bigl(2\max\{1,h\}\bigr)^{\frac{\mu-1}{\gamma}}},
$$
$b=\max\{b_0,b_1\}$ and $c=\min\{c_0,c_2\}.$ Second, following the idea of \cite{portilla2021weighted}, we consider a comparison equation of the form
\begin{gather*}
    \frac{du(t)}{dt}=-\tilde{\rho}u^{\frac{\gamma+\mu-1}{\gamma}}
(t),\\u(0)=u_0=(\alpha_1 + b_2 \Delta^{\mu-1}+b_3\Delta^{\sigma-1})\|\varphi\|_h^\gamma,
\end{gather*}
where $\tilde{\rho}<\rho,$ and
\begin{align*}
1+\tilde{\rho}h&\left(\frac{\mu-1}{\gamma}\right)\\ &\times(\alpha_1 + b_2 \Delta^{\mu-1}+b_3\Delta^{\sigma-1})^{\frac{\mu-1}{\gamma}}\Delta^{\mu-1}\leq \alpha^{\mu-1}.
\end{align*}
The latter condition is necessary, since the lower bound in Lemma~\ref{lemma_lower_bound_functional_pertur_zero_mean_2} holds on the set of functions $S_\alpha$ only.
The solution for the comparison equation is
\begin{equation*}
    u(t) = u_0\left[1+\tilde{\rho}\left(\frac{\mu-1}{\gamma}\right)u_0^{\frac{\mu-1}{\gamma}}t\right]^{-\frac{\gamma}{\mu-1}}.
\end{equation*}
Finally, we arrive at the following result.
\begin{theorem}\label{th:estimation_LK}
Let the trivial solution of system \eqref{eq:system_zero_mean_value} be homogeneous and asymptotically stable. The solutions of system \eqref{eq:system_zero_mean_value} with initial functions satisfying \mbox{$\|\varphi\|_h< \Delta,$} where $\Delta$ is defined in Theorem~\ref{thm:attr_region_LK}, admit an estimate of the form 
\begin{equation*}
    \|x(t,\varphi)\| \leq \hat{c}_1 \|\varphi\|_h\left[1+\hat{c}_2\|\varphi\|_h^{\mu-1}t\right]^{-\frac{1}{\mu-1}},
\end{equation*}
where
\begin{align*}
\hat{c}_1&=\left(\frac{\alpha_1 + b_2 \Delta^{\mu-1}+b_3\Delta^{\sigma-1}}{a_1(\alpha)}\right)^\frac{1}{\gamma}=\frac{\delta}{\Delta},
\\
\hat{c}_2&=\tilde{\rho}\left(\frac{\mu-1}{\gamma}\right)\left(\alpha_1 + b_2 \Delta^{\mu-1}+b_3\Delta^{\sigma-1}\right)^{\frac{\mu-1}{\gamma}}.
\end{align*}
\end{theorem}
\section{Example} 
Consider the system
\begin{align}\label{example_2_perturbed}
&\begin{pmatrix}
\dot{x}_1(t) \\
\dot{x}_2(t) 
\end{pmatrix}
=
\begin{pmatrix}
x_2^\mu(t) \\
-x_1^\mu(t)-x_2^\mu(t-h) 
\end{pmatrix}
\\
    &+\begin{pmatrix}
\cos{t}+\sin(\sqrt{2}t) & 0 \\
0 & \cos{t}+\cos(\sqrt{2}t) 
\end{pmatrix}
\begin{pmatrix}
x_1^\sigma(t-h) \\
x_2^\sigma(t) 
\end{pmatrix},\notag
\end{align}
where $x_1(t),\ x_2(t)\in\mathbb{R},$ $h\in\mathbb{R}^+,$ and $\mu>1$ is a number with odd numerator and denominator. In \cite{book_in_Russian}, the following Lyapunov function for the unperturbed and delay free part of system \eqref{example_2_perturbed} was introduced:
\begin{equation}\label{Lyap_func_example}
    V(x)=\frac{1}{\mu+1}\left(x_1^{\mu+1}+x_2^{\mu+1}\right)+\zeta x_1^\mu x_2,\quad \zeta>0.
\end{equation}
It was shown that if
$$
\zeta<\min\left\{\frac{1}{\mu+1},\frac{4}{(\mu+1)^2}\right\},
$$
then the Lyapunov function \eqref{Lyap_func_example} allows to prove the asymptotic stability of the unperturbed system \eqref{example_2_perturbed} with $h=0.$
Furthermore, the time-derivative of function \eqref{Lyap_func_example} along the trajectories of system \eqref{example_2_perturbed}, when $h=0$, admits an upper bound of the form 
\begin{gather*}
    \dfrac{d V(x(t))}{dt}\leq -\mathrm{w}\|x(t)\|^{2\mu},
    \end{gather*}
where $\mathrm{w}=\dfrac{\kappa}{2^{\mu-1}},$ and
    \begin{gather*}
    \kappa=\min\left\{1-\zeta(\mu+1),\zeta,\frac{\zeta}{1+\zeta}\left(1-\frac{\zeta(1+\mu)^2}{4}\right)\right\}.
\end{gather*}
We take the system parameters $\sigma=\mu=5$, $h=10,$ and set $\zeta=0.0001,$ $\alpha=1.1.$ Compute $\mathrm{w}=6.2\cdot 10^{-6}.$ The constants characterising the Lyapunov function and its derivatives are the following:
\begin{align*}
\alpha_0&=(1/2)^\frac{\mu-1}{2}\left(\frac{1}{\mu+1}-\zeta\right),\\
\alpha_1&=\frac{1}{\mu+1}+\zeta,\\
\beta&=\sqrt{(1+\zeta\mu)^2+(1+\zeta)^2},\quad
\psi=2(\mu+\mu^2\zeta).
\end{align*}
The constants from bounds \eqref{eq:bound_Q} and \eqref{eq:bound_f} are
\begin{gather*}
m_1=m_2=\sqrt{2}, \quad \eta_{11}=\mu,\quad \eta_{12}=0, \quad p_1=p_2=1, \\ q_{11}=q_{22}=\sigma,\quad q_{12}=q_{21}=0.
\end{gather*}
Now, we calculate integral \eqref{eq:L_varepsilon} and compute the function
\begin{equation*}
\omega(\varepsilon)=\varepsilon\left(\dfrac{2\eps+1}{\varepsilon^2+1}
+\dfrac{\textup{max}\{\varepsilon+2\sqrt{2}, 2\varepsilon+\sqrt{2}\}}{\varepsilon^2+2}\right)\xrightarrow[\eps\to 0]{} 0.
\end{equation*}
Choosing the initial function $\varphi(\theta)=[4.9\cdot 10^{-4},\ 4.9\cdot 10^{-4}]^T$, $\theta \in [-10,0]$, we compute the estimate for the solutions of \eqref{example_2_perturbed} from Theorem~\ref{th:estimation_LK}.
The constants characterising the estimate are shown in Table \ref{table_LK_2_perturbed}.
\begin{table}[htb]
\caption{Constants for the estimates of solutions}
\label{table_LK_2_perturbed}
\begin{center}
\begin{tabular}{cccccc}
\hline
$\delta$ & $\Delta$ & $\hat{c}_1$ & $\hat{c}_2$ & $\varepsilon$ & $\tilde{\rho}$\\
\hline
$1\cdot 10^{-3}$ & $7\cdot 10^{-4}$ & $1.41$ & $6.7\cdot 10^{-8}$ & $10^{-14}$ & $3.3\cdot 10^{-7}$\\
 \hline
\end{tabular}
\end{center}
\end{table}
The system response and the computed bound are depicted in Fig.~\ref{figure_estimation_ex_2_perturbed} as continuous and dashed lines, respectively.
    \begin{figure}[H]
     \centering
      \includegraphics[width=0.5\textwidth]{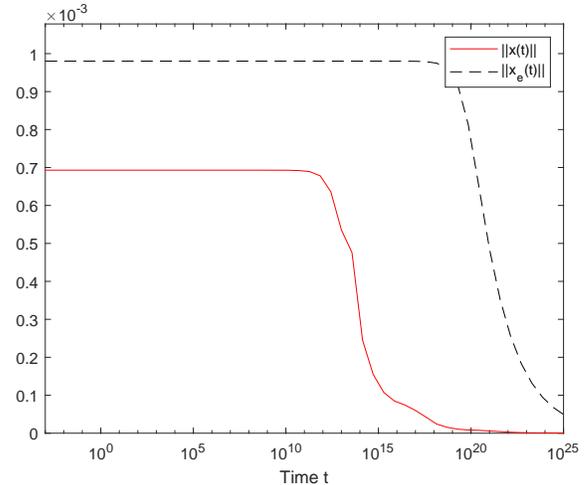}
      \caption{Estimation of the solution of system \eqref{example_2_perturbed} in log scale}
      \label{figure_estimation_ex_2_perturbed}
  \end{figure}
\section{Concluding remarks}
A Lyapunov-Krasovskii functional allowing the analysis of the solutions of homogeneous time-delay systems subject to time varying perturbations allows presenting estimates of the system solutions and of the domain of attraction. The classes of periodic and almost periodic perturbations with zero mean values are studied. An illustrative example validates the results.
\bibliographystyle{IEEEtran.bst}
\bibliography{mybibliography.bib}

\begin{thebibliography}{10}
\providecommand{\url}[1]{#1}
\csname url@rmstyle\endcsname
\providecommand{\newblock}{\relax}
\providecommand{\bibinfo}[2]{#2}
\providecommand\BIBentrySTDinterwordspacing{\spaceskip=0pt\relax}
\providecommand\BIBentryALTinterwordstretchfactor{4}
\providecommand\BIBentryALTinterwordspacing{\spaceskip=\fontdimen2\font plus
\BIBentryALTinterwordstretchfactor\fontdimen3\font minus
  \fontdimen4\font\relax}
\providecommand\BIBforeignlanguage[2]{{%
\expandafter\ifx\csname l@#1\endcsname\relax
\typeout{** WARNING: IEEEtran.bst: No hyphenation pattern has been}%
\typeout{** loaded for the language `#1'. Using the pattern for}%
\typeout{** the default language instead.}%
\else
\language=\csname l@#1\endcsname
\fi
#2}}

\bibitem{aleksandrov2012asymptotic}
A.~Y. Aleksandrov and A.~P. Zhabko, ``On the asymptotic stability of solutions
  of nonlinear systems with delay,'' \emph{Siberian Mathematical Journal},
  vol.~53, no.~3, pp. 393--403, 2012.

\bibitem{aleksandrov2014delay}
A.~Y. Aleksandrov, G.-D. Hu, and A.~P. Zhabko, ``Delay-independent stability
  conditions for some classes of nonlinear systems,'' \emph{IEEE Transactions
  on Automatic Control}, vol.~59, no.~8, pp. 2209--2214, 2014.

\bibitem{efimov2016}
D.~Efimov, A.~Polyakov, W.~Perruquetti, and J.-P. Richard, ``Weighted
  homogeneity for time-delay systems: finite-time and independent of delay
  stability,'' \emph{IEEE Transactions on Automatic Control}, vol.~61, no.~1,
  pp. 210--215, 2016.

\bibitem{rosier1992homogeneous}
L.~Rosier, ``Homogeneous {L}yapunov function for homogeneous continuous vector
  field,'' \emph{Systems and Control Letters}, vol.~19, no.~6, pp. 467--473,
  1992.

\bibitem{hermes1991homogeneous}
H.~Hermes, ``Homogeneous coordinates and continuous asymptotically stabilizing
  feedback controls,'' \emph{Differential Equations, Stability and Control},
  vol. 109, no.~1, pp. 249--260, 1991.

\bibitem{aleksandrov2016asymptotic}
A.~Y. Aleksandrov, E.~Aleksandrova, and A.~P. Zhabko, ``Asymptotic stability
  conditions and estimates of solutions for nonlinear multiconnected time-delay
  systems,'' \emph{Circuits, Systems, and Signal Processing}, vol.~35, no.~10,
  pp. 3531--3554, 2016.

\bibitem{aleksandrov1996stability}
A.~Y. Aleksandrov, ``The stability of equilibrium of non-stationary systems,''
  \emph{Journal of Applied Mathematics and Mechanics}, vol.~60, no.~2, pp.
  199--203, 1996.

\bibitem{Voronezh}
A.~Y. Aleksandrov, A.~P. Zhabko, and V.~Pecherskiy, ``Complete type functionals
  for some classes of homogeneous differential-difference systems,''
  \emph{Proc. 8th International Conference ``Modern methods of applied
  mathematics, control theory and computer technology''}, pp. 5--8 (in
  Russian), 2015.

\bibitem{alexandrova2019lyapunov}
A.~P. Zhabko and I.~V. Alexandrova, ``Complete type functionals for homogeneous
  time delay systems,'' \emph{Automatica}, vol. 125, p. 109456, 2021.

\bibitem{kharitonov2013time}
V.~Kharitonov, \emph{Time-delay systems: {L}yapunov functionals and
  matrices}.\hskip 1em plus 0.5em minus 0.4em\relax Basel: Birkh\"auser, 2013.

\bibitem{zubov1964methods}
V.~I. Zubov, \emph{Methods of A.M. Lyapunov and their application}.\hskip 1em
  plus 0.5em minus 0.4em\relax P. Noordhoff, 1964.

\bibitem{alexandrova2018junction}
I.~V. Alexandrova and A.~P. Zhabko, ``At the junction of
  {L}yapunov-{K}rasovskii and {R}azumikhin approaches,''
  \emph{IFAC-PapersOnLine}, vol.~51, no.~14, pp. 147--152, 2018.

\bibitem{fink2006almost}
A.~M. Fink, \emph{Almost periodic differential equations}.\hskip 1em plus 0.5em
  minus 0.4em\relax Springer, 1974.

\bibitem{Krasovskii}
N.~N. Krasovskii, \emph{Certain Problems of Stability Theory of Motion}.\hskip
  1em plus 0.5em minus 0.4em\relax Moscow: Fizmatgiz, 1959. In Russian.

\bibitem{melchor2007estimates}
D.~Melchor-Aguilar and S.-I. Niculescu, ``Estimates of the attraction region
  for a class of nonlinear time-delay systems,'' \emph{IMA J. Math. Control
  Information}, vol.~24, no.~4, pp. 523--550, 2007.

\bibitem{portilla2021weighted}
G.~Portilla, I.~V. Alexandrova, and S.~Mondi{\'e}, ``Estimates for weighted
  homogeneous delay systems: a {L}yapunov-{K}rasovskii-{R}azumikhin approach,''
  \emph{American Control Conference, Accepted}, 2021.

\bibitem{portilla2020comparison}
G.~Portilla, I.~V. Alexandrova, S.~Mondi{\'e}, and A.~P. Zhabko, ``Estimates
  for solutions of homogeneous time-delay systems: comparison of
  {L}yapunov--{K}rasovskii and {L}yapunov--{R}azumikhin techniques,''
  \emph{International Journal of Control, Submitted}, 2021.

\bibitem{book_in_Russian}
A.~Y. Aleksandrov, E.~B. Aleksandrova, A.~V. Ekimov, and N.~V. Smirnov, \emph{A
  collection of tasks and exercises on the theory of stability: a
  handbook}.\hskip 1em plus 0.5em minus 0.4em\relax Publishing House "Lan",
  2016. In Russian.

\end{thebibliography}
\end{document}